\documentclass[accepted=2020-04-14,a4paper,onecolumn]{quantumarticle}

\pdfoutput=1

\usepackage{amsmath,amsthm,amssymb,amsfonts,setspace}
\usepackage{bm}
\usepackage{float}
\usepackage{color}
\usepackage[usenames,dvipsnames]{xcolor}
\usepackage[colorlinks]{hyperref}
\usepackage{mathtools}
\usepackage{cleveref} 

\usepackage[numbers,sort&compress]{natbib}

\usepackage{graphicx}
\graphicspath{ {./figs/}} 

\crefformat{equation}{Eq.~(#2#1#3)} 
\crefformat{section}{Section~#2#1#3}
\crefformat{theorem}{Theorem~#2#1#3}
\crefformat{lemma}{Lemma~#2#1#3}
\crefformat{appendix}{Appendix~#2#1#3}
\Crefformat{equation}{Equation~(#2#1#3)}
\crefformat{figure}{Fig.~#2#1#3}
\crefformat{corollary}{Corollary~#2#1#3}
\crefrangeformat{equation}{Eqs.~#3(#1)#4--#5(#2)#6}

\newcommand{\bra}[1]{\langle{#1}\vert} 
\newcommand{\ket}[1]{\vert{#1}\rangle} 
\newcommand{\ketbra}[2]{\vert{#1}\rangle\!\langle{#2}\vert}

\newcommand{\Per}{\textrm{Per}}

\renewcommand{\L}{\left(} 
\newcommand{\R}{\right)} 

\newcommand{\dg}{^\dagger}

\newcommand{\tr}{\mathrm{tr}}

\renewcommand{\o}{\textrm{o}} 
\renewcommand{\O}{\textrm{O}} 

\newcommand{\N}{\mathcal{N}} 
\newcommand{\M}{\mathcal{M}} 
  
\newcommand{\tran}{\eta} 

\newcommand{\eqdef}{\mathrel{:=}}
\newcommand{\proj}[2]{\mathbb{P}_{\textrm{sym}}^{#1,#2}}

\theoremstyle{plain}
\newtheorem{theorem}{Theorem}
\theoremstyle{plain}
\newtheorem{exa}{Example}
\theoremstyle{plain}
\newtheorem{corollary}{Corollary} 
\theoremstyle{plain}

\theoremstyle{plain}

\theoremstyle{plain}
\newtheorem{remark}{Remark}
\theoremstyle{plain}

\newcommand{\daniel}[1]{{\color{black} #1}}

\newcommand{\michal}[1]{{\color{black} #1}}

\begin{document}
\title{Classical simulation of linear optics subject to nonuniform losses}

\author{Daniel J.\ Brod}
\email{danieljostbrod@id.uff.br}
\affiliation{Instituto de F\'{\i}sica, Universidade Federal Fluminense, Niter\'oi, RJ, 24210-340, Brazil}

\author{Micha\l\ Oszmaniec}
\email{michal.oszmaniec@gmail.com}
\affiliation{International Centre for Theory of Quantum Technologies, University of Gdansk, Wita Stwosza 63, 80-308 Gdansk, Poland}
\affiliation{Center for Theoretical Physics, Polish Academy of Sciences, Al. Lotników 32/46, 02-668
Warszawa, Poland}

\date{April 14th, 2020}

\begin{abstract}
We present a comprehensive study of the impact of non-uniform, i.e.\ path-dependent, photonic losses on the computational complexity of linear-optical processes. Our main result states that, if each beam splitter in a network induces some loss probability, non-uniform network designs cannot circumvent the efficient classical simulations based on losses.

To achieve our result we obtain new intermediate results that can be of independent interest. First we show that, for any network of lossy beam-splitters, it is possible to extract a layer of non-uniform losses that depends on the network geometry. We prove that, for every input mode of the network it is possible to commute $s_i$ layers of losses to the input, where $s_i$ is the length of the shortest path connecting the $i$th input to any output. We then extend a recent classical simulation algorithm due to {P.\ Clifford and R.\ Clifford} to allow for arbitrary $n$-photon input Fock states (i.e. to include collision states).  Consequently, we identify two types of input states where boson sampling becomes classically simulable: (A) when $n$ input photons occupy a constant number of input modes; (B) when all but $O(\log n)$ photons are concentrated on a single input mode, while an additional $O(\log n)$ modes contain one photon each.

\end{abstract}

\maketitle

\section{Introduction and relation to previous works}

The recent paradigm of quantum computational advantage (or supremacy), is regarded as a promising route towards demonstrating that quantum computers are more powerful than their classical counterparts \cite{Harrow2017}. Although the systems used in this approach are not expected to be universal for quantum computation, and are often not known to perform explicitly useful tasks, they allow us to test quantum mechanics in the limit of high computational complexity. The pursuit of a near-term demonstration of quantum {superiority} is also aligned with the overarching goals of the field of quantum computing, as it pushes the development of technologies that will undoubtedly be necessary for scalable universal quantum computers, and allows us to better understand the effects of real-world imperfections in intermediate-size quantum systems. 

One candidate for demonstrating quantum advantage is nonadaptive linear optics, or boson sampling \cite{Aaronson2013a}. In this model, an $n$-photon $m$-mode Fock state evolves according to a passive $m$-mode linear-optical transformation $U$ and is subsequently measured by particle-number resolving detectors [see \cref{fig:inputs}(a)]. This physical process leads to a probability distribution $p_U (\vec{T})$. The seminal contribution of \cite{Aaronson2013a} was to show that, subject to certain complexity-theoretic conjectures, for \emph{generic} transformations $U$ it is hard for a classical computer to generate samples from a distribution suitably close to $p_U (\vec{T})$. Besides being an elegant and physically-motivated computational model and, arguably, the first proposal of quantum supremacy as we understand it today, boson sampling also benefits from the technology and expertise developed by the quantum optics community in the last decades \cite{RevPhot2019}.

Since the original proposal several small-scale boson sampling experiments have been reported \cite{Broome2013,Crespi2013b,Spring2013,Tillmann2013,Spagnolo2014,Carolan2014,Carolan2015,Bentivegna2015,Loredo2017,He2017,Wang2017,Wang2018}, with state-of-the-art ranging up to five photons with near-deterministic quantum dot sources \cite{Loredo2017,He2017,Wang2017}. From the opposite end, there has also been intense activity in the development of classical simulation algorithms \cite{Neville2017,Clifford2017,Roga2019}, and current supercomputers are expected to simulate 50-photon experiments without much difficulty \cite{Wu2016,Clifford2017}. A more refined  analysis of the complexity-theoretic arguments underpinning boson sampling has suggested 90 photons as a concrete milestone for the demonstration of quantum computational advantage \cite{Dalzell2018}.

Many technological and theoretical challenges remain that can undermine the scalability of boson sampling. Several sources of imperfection affect linear-optical experiments, and it is essential to understand which can be mitigated and which degrade the computational power of the model. This boundary between classical simulability and quantum advantage is an area of intense investigation, with losses in particular receiving most attention. In \cite{Aaronson2016} it was shown that boson sampling retains its computational power if a constant number of photons is lost. In the other extreme, recent papers showed that lossy boson sampling can be efficiently simulated when less than $\sqrt{n}$ photons are left \cite{GarciaPatron2017,Oszmaniec2018} for arbitrary interferometers, or even when a suitably high  \emph{constant fraction} of the photons is lost for typical Haar-random interferometers \cite{Renema2018} \michal{(see also a recent follow-up \cite{Moylett2019})}. Boson sampling has also been investigated under the effect of fabrication noise in the linear-optical components \cite{Arkhipov2015,Leverrier2014}, losses combined with dark counts \cite{Rahimi-Keshari2016}, partial photon distinguishability \cite{Renema2017}, and Gaussian noise in the experimental data \cite{Kalai2014}.  

In this work we build on our previous findings of \cite{Oszmaniec2018} and solve some of the open problems posed there. The main result of \cite{Oszmaniec2018} stated that, when less than $\sqrt{n}$ out of $n$ photons are left, it is possible to approximate a lossy boson sampling state by a state of distinguishable photons, which is known to be classically simulable. This approximation is done at the level of the \emph{input} state, and holds for arbitrary linear-optical experiments. For this result to work, it was necessary to assume losses happen at the input--a reasonable assumption when the interferometer is balanced, but which might fail for unbalanced network geometries \cite{Reck1994}. To justify this, in \cite{Oszmaniec2018} we also showed it is possible to commute $s$ layers of mode-independent losses to the input of the network, where $s$ is the smallest number of beam splitters in any path connecting any input to any output. However, it is possible to construct a universal linear-optical transformation where $s=1$ \cite{Reck1994}, suggesting a possible route to circumvent the algorithm of \cite{Oszmaniec2018}.

Here, we extend the results of \cite{Oszmaniec2018} in different directions, with the goal of adapting them for very unbalanced network geometries \cite{Reck1994}. In particular, our main result is that \michal{an} unbalanced network \daniel{geometry} cannot be used to circumvent classical simulation due to losses. Our intermediate results regarding the accumulation of losses in linear-optical networks and \michal{an extension of the} algorithm for simulation of boson sampling \cite{Clifford2017} can be also be of independent interest to the quantum optics and quantum complexity theory communities. 

The paper is organized as follows. First, in \cref{sec:inf} we give an informal overview of our results. Sections \ref{sec:result1}-\ref{sec:result4} contain formal statements and proofs of our main results. Finally, \cref{app:margprobs} contains the technical details needed in the proof of \cref{thm:binnedBS}.

{\emph{Note added.} After completion of this work, the first experimental demonstration of quantum computational advantage was realised in  a chip consisting of 53  superconducting qubits \cite{Suprem2019}. The experiment implemented random circuit sampling, which constitutes another paradigm for realising quantum computational superiority \cite{Boixo2016}. Around the same time, a boson sampling experiment with 20 photons occupying 60 optical modes was reported in \cite{BosNew2019}.}

\begin{figure}[t]
    \centering
    \includegraphics[width=0.7\textwidth]{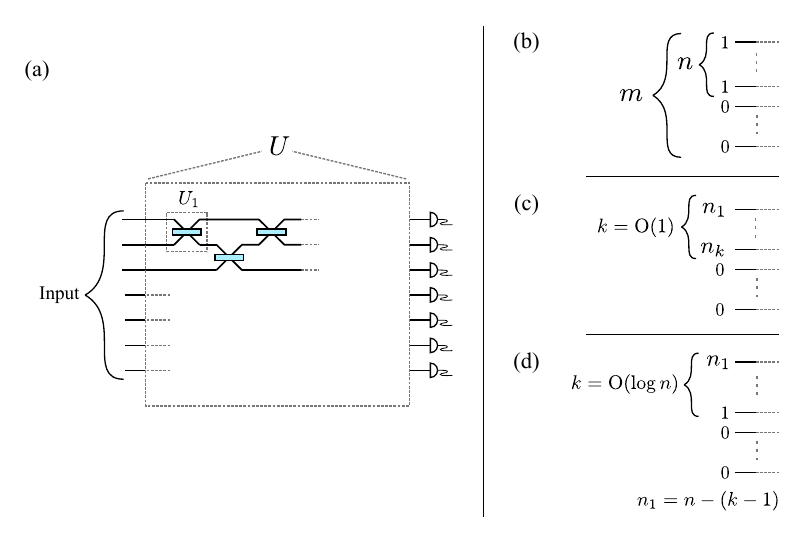}
    \caption{The boson sampling setup (in absence of losses). (a) An $n$-photon state enters a linear optical network described by $m\times m$ unitary matrix $U$ and is subsequently measured by particle-number resolving detectors. (b) Standard boson sampling input, where the first $n$ modes are occupied by individual photons. (c) Input of type A, where $n$ photons are equally distributed in a constant number of input modes. (d) Input of type B, where $O(\log(n))$ \daniel{modes} are occupied by single photons and the remaining photons occupy a single mode.}
    \label{fig:inputs}
\end{figure}

\textbf{Notation and main concepts.}  Throughout the paper we use the following notation. Given two positive-valued functions $f$ and $g$, we write  $f=\o(g)$ if $\lim_{x\rightarrow\infty} f(x)/g(x) = 0$ and $f=\O(g)$ if $ \lim_{x\rightarrow\infty} f(x)/g(x) < \infty$. Likewise, $g=\omega(f)$ if $f=\o(g)$. Finally, $f\approx g$ when $ \lim_{x\rightarrow\infty} f(x)/g(x)=1$. By $\lfloor x \rfloor$  ($\lceil x \rceil$) we denote the largest (smallest) integer smaller (larger) or equal to $x$. 
We denote the total number of particles by $n$, the total number of modes by $m$, and by $[m]$ the $m$-element set $\lbrace 1,2,\ldots, m\rbrace$. We alternatively refer to input modes as ``bins'', whenever there are many photons initialized in that mode (accordingly, we refer to this situation as bunching or binning). \daniel{We also refer to any state where more than one photon occupies the same mode as a collision state}.

The Hilbert space of photonic states in $m$ modes is spanned by Fock states, characterized by a definite photon number $s_i$ in each optical mode $i$. In what follows we denote Fock states by  $\ket{\vec{S}}$, where $\vec{S}=(s_1,\ldots,s_m)$ is a vector of occupation numbers. In the formalism of second quantization we write $\ket{\vec{S}}=\prod_{i=1}^m \frac{(a^\dagger_i)^{s_i}}{\sqrt{s_i!}}\ket{0}$, where $a^\dagger_i$ is the standard creation operator associated to bosonic mode $i$ and $\ket{0}$ is the vacuum state. Passive linear optical transformations are associated with $m\times m$ unitary matrices $U$, which induce the transformation of modes $ a^\dagger_i \mapsto a'^\dagger_i = \sum_{j=1}^m U_{ji} a^\dagger_i$. When this transformation is applied to input state $\ket{\vec{S}}$ the output probabilities are given by
\begin{equation}\label{eq:statistics}
p_U\left(\vec{S}\rightarrow\vec{T}\right)=\frac{\left|\mathrm{Per}(U_{ST})\right|^2}{\prod_{i=1}^m s_i! \prod_{i=1}^m t_i!}\ ,
\end{equation}
where $\mathrm{Per}(A)$ denotes the permanent of $A$ and $U_{ST}$ is an $n\times n$ matrix constructed according to the following prescription. We first construct an $m \times n$ matrix $U_S$ by taking $s_i$ copies of the $i$th column of $U$, then construct an $n \times n$ matrix $U_{S,T}$ by taking $t_j$ copies of the $j$th row of $U_T$. \footnote{{The arguments for computational hardness of Boson sampling rely on the fact that computation of the permanent is a $\#\mathrm{P}$-hard problem. Importantly, however, a device implementing boson sampling does not, in general, have the power to accurately compute or estimate permanents that appear in Eq.\ \eqref{eq:statistics}.}}.  Boson sampling can also be framed in the language of first quantization, which we briefly use (see e.g. \cite{Moylett2018} or \ the introductory section of \cite{Oszmaniec2018} for the translation between the two descriptions).

\section{Overview of main results}\label{sec:inf} 

In this Section we give an informal summary of our main findings and discuss the relation to previous work.

\textbf{Result 1: Extracting nonuniform losses from the network.} We improve the procedure of \cite{Oszmaniec2018} to commute losses to the input of a network. For simplicity, we make the following assumptions regarding losses. We assume that all losses happen within the network and are located at  beam splitters, and that every beam splitter carries the same amount of loss. Although losses due to imperfect sources and detectors are experimentally important \cite{Wang2017}, they are in principle constant, whereas losses due to linear-optical elements dictate how overall transmissivity of the network scales as the experiments become more complex. Furthermore, the geometry of a network is the main source of loss non-uniformity which we wish to study. Experimentally, beam splitter losses are also relevant e.g.\ in integrated photonic devices, since waveguide bends required to build directional couplers can cause photons to leave in unguided modes \cite{Crespi2013b}. We represent a lossy linear-optical \michal{network} as in \cref{fig:losspullout}(a).

\begin{figure}[t]
    \centering
    \includegraphics[width=0.8\textwidth]{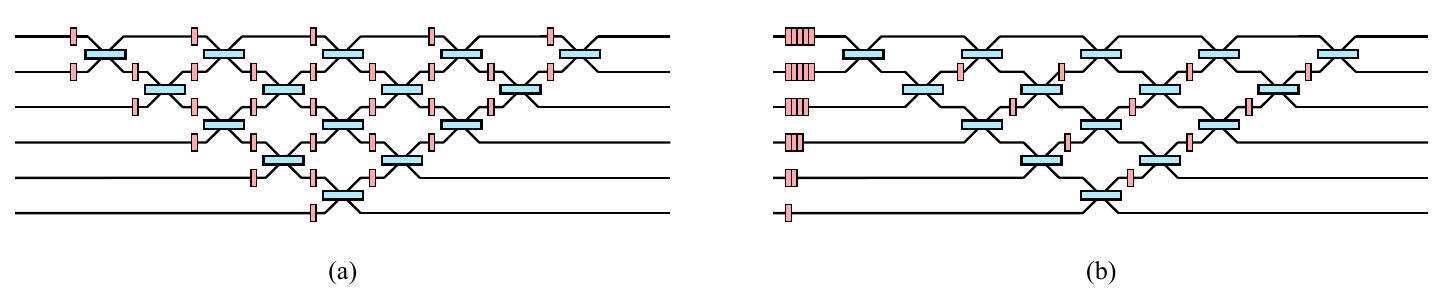}
    \caption{Extraction of losses from an unbalanced photonic network. Blue rectangles correspond to beam splitters and associated phase shifters, while orange rectangles are loss elements. We assume all loss elements are equal, and use the fact that two equal loss elements in the arms of a beam splitter commute with it. We apply Result 1 to the triangular decomposition of Reck \emph{et al} \cite{Reck1994} in (a) as an example. It is possible to commute to the input of the network different rounds of losses according to the shortest path between that input and any output, leading to the equivalent network represented in (b).}
    \label{fig:losspullout}
\end{figure}

The main figure of merit for the result of \cite{Oszmaniec2018} is the \emph{length} of input-output paths inside a network, measured in the number of beam splitters a photon has to traverse along that path. Previously it was known that, if the shortest input-output path has length $s$, we can rewrite the network in such a way that it has $s$ layers of mode-independent losses at the input. However this is not particularly useful for unbalanced networks such as that of Reck \emph{et al} \cite{Reck1994}, depicted in \cref{fig:losspullout}(a), as in that case we can only pull out one layer of uniform losses. 

Our new result improves on this by showing how to extract \emph{nonuniform} loss layers. More specifically, let $s_i$ be the shortest path between input $i$ and any output within some network. We show how to write an equivalent network that is preceded by a round of losses where input $i$ suffers the effect of $s_i$ consecutive loss elements. The precise formulation of this result can be found in \cref{thm:nonuniform}.  Our proof is efficient and constructive, and works for networks of arbitrary geometry. The final result is illustrated in \cref{fig:losspullout}(b) for the case of the network of Reck \cite{Reck1994}. Overall, our result gives a more efficient description of the network where asymmetry is taken into account, and which we use to prove result 4 described later. 

\textbf{Result 2: Classical simulation of boson sampling with collision input states.} It is well known that boson sampling is efficiently classically simulable if all $n$ photons are initialized in a single mode, i.e., input state $\ket{n, 0, 0, \ldots, 0}$. This follows from the fact that the permanent of a matrix of $n$ repeated columns can be computed trivially in linear time \cite{Aaronson2013a}, or alternatively from the fact that if all photons are initialized in the same mode they behave as distinguishable particles \cite{Oszmaniec2018}. Here we prove a stronger version of this result, by identifying two new types of input state which allow for efficient classical simulation of boson sampling [see \cref{fig:inputs}(b) and \cref{fig:inputs}(c)]. The first (type A) corresponds to $n$ photons concentrated in a constant number of input modes. The second (type B) is when all but $O(\log n) $ photons are concentrated on a single input mode, while an additional $O(\log n)$ modes contain one photon each.

Concretely, we provide a strong simulation of boson sampling (i.e.\ computing outcome probabilities) by showing how an expression for the permanent function from \cite{Chin2017} can be computed efficiently for inputs of types A and B (though a similar scaling follows from an Appendix of \cite{Valery2013}). We then {generalize} a result of \cite{Clifford2017} for weak classical simulation (i.e.\ producing samples from the correct distribution) to allow for bunched inputs, showing that it is efficient for inputs of types A and B. Our {extension} of the result of \cite{Clifford2017} is also based on a physical description of the state (via first quantization) rather than combinatorial considerations, and we believe it could make the original result more transparent for researchers with a stronger physics background. The technical formulation of this result can be found in \cref{thm:binnedBS} and in \cref{cor:TwoTypesClass}.
 
Besides giving us a new understanding of the regimes where efficient classical simulability of boson sampling is possible, and ruling out attempts at demonstrating a quantum advantage with highly concentrated bosonic input states, this result also serves as a nontrivial stepping stone towards result 3 described below.

\textbf{Result 3: Classical simulation of lossy boson sampling \michal{with a few}  lossless modes}. We improve the classical simulation algorithm of \cite{Oszmaniec2018} by allowing the lossy boson sampling instance to be coupled to some lossless modes. Concretely, suppose that $n$ photons pass through a lossy channel in such a way that less than $\sqrt{n}$ survive, and are then combined with up to $c \log n$ photons input in modes that have perfect transmission. We show that the classical simulation based on the ideas similar those of \cite{Oszmaniec2018} works for this setting as well. The technical formulation of this result can be found in \cref{thm:notLOSTclass}.

The classical simulation of \cite{Oszmaniec2018} was based on approximating a lossy bosonic state $\rho$ by another state $\sigma$ where photons effectively behave as distinguishable particles, then leveraging known results for classical simulation in that case \cite{Aaronson2013a}. Our new result is based on appending up to $c \log n$ single-photon states to $\sigma$, and using this state to simulate the state obtained by also appending $c \log n$ photons to $\rho$. The state $\sigma$ is a convex combination of states of type B so we can leverage result 2 above. The distance between the two states (and hence the error bound on the simulation) is the same as in \cite{Oszmaniec2018}, but the simulation is more costly due to the use of \cref{thm:binnedBS}. This result is also essential to prove result 4 that follows.

\textbf{Result 4: Nonuniform losses do not avoid classical simulation.} By combining results 1--3 we obtain our main result. Informally it states that the classical simulation algorithms of \cite{Oszmaniec2018,GarciaPatron2017} cannot be circumvented by the design of unbalanced networks such as that of Reck \emph{et al} \cite{Reck1994} (see \cref{thm:notUNIFclass} for the formal statement). This was left as an open question in \cite{Oszmaniec2018}. The idea was that, in an unbalanced network, there might be some input modes where losses are very mild, and the high chance of survival of these photons might break the assumptions behind the simulation algorithm in \cite{Oszmaniec2018,GarciaPatron2017}. Here we show that classical simulation remains possible as long as only $O\log(n)$ inputs have \emph{small} optical path to the output. Our result can be immediately applied to lossy interferometers with the geometry proposed in \cite{Reck1994}.  This closes one avenue for scaling up boson sampling experiments, but is also relevant for current experiments since this geometry is common in integrated-photonic implementations \cite{Carolan2015,Crespi2013b}.

In what follows we outline our argument for the simulability of the triangular construction of \cite{Reck1994} [cf.\ \cref{fig:losspullout}(b)]. This construction has the property that the shortest path from mode $i$ to any output has length $i$ (if we label modes starting from the lowest one). Suppose the input state we wish to simulate occupies the bottom $n$ modes in an attempt to avoid losses as much as possible. To simplify the argument, suppose now that the bottom $c \log n$ modes are lossless (or include their loss channel as part of the network rather than the input state). By result 1 we can extract more than $c \log n$ loss elements for all modes above mode $c \log n + 1$. However, as discussed in \cite{Oszmaniec2018}, given a per-beam splitter loss parameter $\eta$, there is always some $c$ such that, if an $n$-photon state suffers $c \log n$ rounds of losses, we expect less than $\sqrt{n}$ photons to survive on average. Therefore, the input to the network is amenable to approximation by a convex combination of \daniel{states} of type B, as described in result 3, and so efficient classical simulation is possible.

\section{Extracting nonuniform losses from a lossy network} \label{sec:result1}


Linear-optical networks are the natural implementation of general multimode interferometers. They are typically built by composition of simple few-mode transformations, often easier to realize in practice, which are then arranged in some network layout (see \cite{Reck1994,Carolan2014} for the commonly used geometries, \cite{Bouland2014,Sawicki2016,SawickiKarnas2017} for  mathematical perspective, and \cite{Slowik2018} for the generalization to non-unitary transformations). At the same time, each small transformation also introduces some loss. The purpose of this Section is to provide an effective description of lossy linear-optical networks that takes into account  their geometric properties.  

Our focus is on losses induced by linear-optical elements within the circuit, such as beam-splitters, which determine how losses \emph{scale} with the size of experiments. We disregard losses in detectors and sources, as well as uniform losses associated with free-space or waveguide propagation. These do not depend on the network layout and can be easily incorporated into our description as additional losses at the input of the network \daniel{(alternatively, any path-dependent transmission loss can be incorporated into the nearest beam splitter)}. A loss channel that acts uniformly on $m$ modes is parameterized by a transmission probability $\tran$, and its action on an $n$-photon state can be written as:
\begin{equation}\label{eq:BSuniformMODEL}
\rho \mapsto \Lambda_\tran \L \rho \R \eqdef \sum_{l=0}^n \binom{n}{l} \tran ^l (1-\tran)^{n-l} \tr_{n-l} \L \rho \R ,
\end{equation}
where $\binom{n}{l} \tran^l (1-\tran)^{n-l}$ is the probability that $l$ out of $n$ particles remain. This is equivalent to adding an extra beam splitter that, with amplitude $\sqrt{1-\tran}$, diverts each photon into some inaccessible mode [see \cref{fig:bslosses}(a)]. It can be easily shown that $\Lambda_{\tran_2} \circ \Lambda_{\tran_1} = \Lambda_{\tran_1 \tran_2}$.	

\begin{figure}[h]
    \centering
    \includegraphics[width=0.7\textwidth]{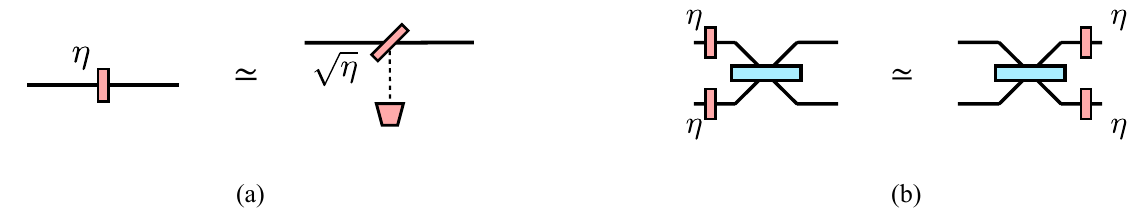}
    \caption{(a) A loss element with transmission probability $\tran$, which we represent by an orange rectangle, can be modeled by a beam splitter that, with amplitude $\sqrt{1-\tran}$, reroutes each photon into an inaccessible mode. (b) We assume each beam splitter in the network is accompanied by two identical loss elements at its inputs. Uniform losses commute with linear optics, a property we apply repeatedly at the level of beam splitters to deal with nonuniform losses in general unbalanced networks.}
    \label{fig:bslosses}
\end{figure}

It is well-known that a layer of identical loss elements commutes with any linear-optical network acting only on the same modes. In \cref{fig:bslosses}(b) this is represented for $m=2$. This is the rationale for a standard assumption that losses happen at the input of a network, as it is usually assumed that they are approximately mode-independent. In practice, however, losses can also occur in nontrivial optical elements composing a network. In what follows,  we assume for simplicity that each beam splitter carries two identical loss elements at its inputs.  In \cite{Oszmaniec2018} we showed the following. Given some linear-optical network built out of the building block of \cref{fig:bslosses}(b), let $s$ be the smallest number of beam splitters that a photon has to traverse within the network starting from any input to reach any output. It is possible to rewrite this network as a different (possibly still containing some losses) linear-optical network preceded by $s$ uniform layers of loss elements. This formalized the assumption that losses happen at the input by basing it on a geometrical property of the network, namely the smallest number of loss elements that will affect any photon as it propagates inside it.

The above result works very well for symmetric networks, such as that of Clements \emph{et al} \cite{Clements2016}, depicted in \cref{fig:circuits}(a), but not for unbalanced network such as the triangular decomposition of Reck \emph{et al} \cite{Reck1994}, depicted in \cref{fig:circuits}(b) (often used in integrated photonics \cite{Carolan2015,Crespi2013b}). For the construction of Reck \emph{et al} specifically, the previous result only allows commuting a single layer of losses to the input. This led us to conjecture it might be possible to circumvent the efficient classical simulation of \cite{Oszmaniec2018} by using purposely unbalanced networks. Motivated by this we now prove a stronger result that allows us to treat each mode individually. Namely we show that for, each input mode $i$, it is possible to pull $s_i$ loss elements to the beginning of the network, where $s_i$ is the smallest number of beam splitters in any path between mode $i$ and any output mode.

\begin{figure}
    \centering
    \includegraphics[width=0.8\textwidth]{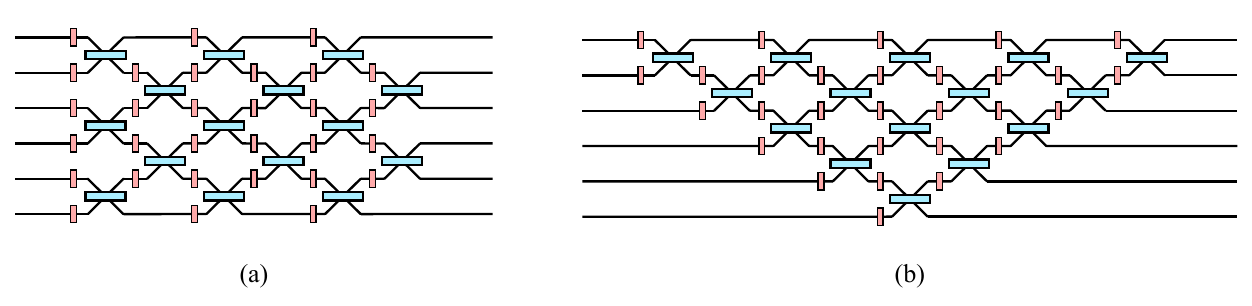}
    \caption{Two decompositions of universal linear optics on 6 modes. (a) The symmetric decomposition of Clements \emph{et al} \cite{Clements2016}. (b) The triangular decomposition of Reck \emph{et al} \cite{Reck1994}. The triangular decomposition displays very nonuniform losses casting doubt on the standard assumption that losses happen at the input. Orange rectangles represent losses, not phase shifters, as common in similar figures in the literature. We do not represent phase shifters explicitly.}
    \label{fig:circuits}
\end{figure}{}

\begin{theorem}\label{thm:nonuniform}
Let $\N$ be an $m$-mode linear-optical network composed of beam splitters, each associated with a transmission probability $\eta$.  Let the \emph{length} of a path connecting some input and output modes be the number of beam splitters that a photon encounters along that path. Let $s_i$ be the shortest path between input mode $i$ and any output mode. The quantum channel associated with the entire network can be decomposed as follows:
\begin{equation}\label{eq:nonBALdec}
\Lambda_{\N} = \Lambda_{\tilde{\N}} \circ \Lambda_{\vec{\tran}}\ .
\end{equation}
Here $\Lambda_{\vec{\tran}}$ is a nonuniform loss channel described by $\vec{\tran} = (\tran_1,\tran_2, \ldots , \tran_m)$, where $\tran_i=\tran^{s_i}$ is the transmission probability for mode $i$, and $\Lambda_{\tilde{\N}}$ is another (possibly still lossy) linear-optical network that can be described explicitly and efficiently \daniel{(in the number of modes and beam splitters in the network)}.
\end{theorem}

\begin{proof} We assume without loss of generality that the network is arranged in $D$ layers, where each layer contains at most $\left \lfloor{m/2}\right \rfloor$ beam splitters acting on disjoint pairs of modes. Under our assumption that each beam splitter is preceded by one loss element in each of its inputs, the length of some input-output path corresponds to the number of loss elements the photon encounters in that path.

We prove the Theorem by iteratively applying the identity of \cref{fig:bslosses}(b). We begin with the \emph{last} layer in the network, and build it by progressively adding each previous layer one beam splitter at a time. At each step we commute losses through the new beam splitter to make sure the new network has a specific form.	

Consider initially the more general situation of \cref{fig:iteration}(a). We have some linear-optical network $\tilde{\M}$. It is composed of a network $\M$, that can contain any number of beam splitters and loss elements distributed in an arbitrary manner, preceded by a layer of loss elements with transmissivity $\mu_i$ in mode $i$. We add to the beginning of $\M$ a new beam splitter, $\mathcal{B}$, acting on modes $i$ and $j$ (and two corresponding loss elements with transmissivity $\tran$). We cannot necessarily commute $\mu_i$ and $\mu_j$ through $\mathcal{B}$, as they may not be equal, but we can commute losses corresponding to the smallest of them. Suppose $\mu_i \leq \mu_j$, and write $\mu_j = (\mu_j/\mu_i) \mu_i$. We can commute two loss elements $\mu_i$ through $\mathcal{B}$, leaving behind a loss element $\mu_j/\mu_i$ in mode $j$. We are now left with the situation of \cref{fig:iteration}(b).

\begin{figure}
    \centering
    \includegraphics[width=0.7\textwidth]{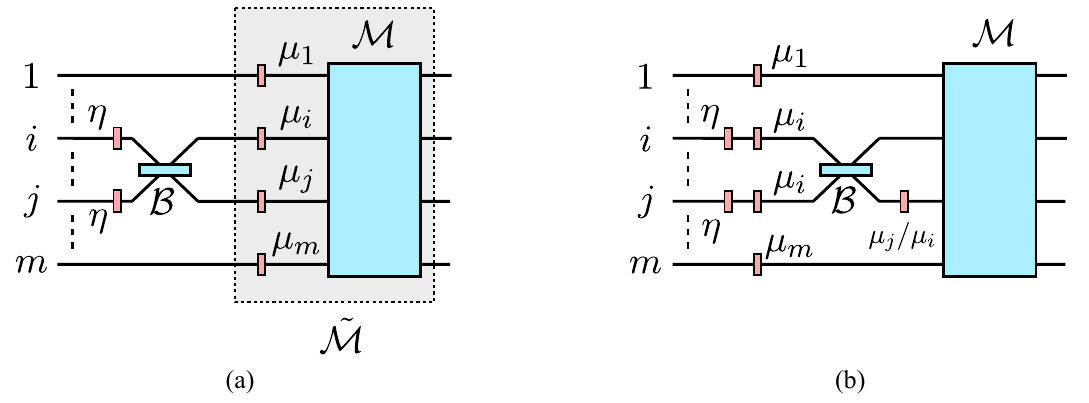}
    \caption{(a) A network $\tilde{\M}$ composed of some arbitrary (possibly lossy) linear-optical network $\M$ preceded by a round of nonuniform losses. Beam splitter $\mathcal{B}$, acting on modes $i$ and $j$, precedes $\tilde{\M}$. (b) Losses at the input of $\tilde{\M}$ have been commuted through $\mathcal{B}$. Since $\mu_i$ and $\mu_j$ might not be equal, only losses corresponding to the smallest of them (in this figure, $\mu_i$) commute through $\mathcal{B}$. \cref{thm:nonuniform} iterates this procedure \michal{to} extract the greatest number of loss elements possible to the input.}
    \label{fig:iteration}
\end{figure}{}

Let us now return to our network $\N$ and apply the above procedure. Start with the last layer of beam splitters. Clearly it can be written in the form of $\tilde{\M}$ of \cref{fig:iteration}(a). We now add each previous layer, one beam splitter at a time, always applying the argument of the previous paragraph and \cref{fig:iteration}. Suppose that, at step $k$ of this iteration, we have some intermediate network $\tilde{\M}_k$ as in \cref{fig:iteration}(a). Suppose also that $\mu_i = \tran^{t_i}$ where $t_i$ is the length of the shortest path between input mode $i$ of $\M_k$ and any output. In other words, we assume that at step $k$ the intermediate network satisfies the statement of our theorem. 

Consider now the next beam splitter to be added to the network, $\mathcal{B}$, acting between modes $i$ and $j$. By the reasoning of \cref{fig:iteration} we can now commute $\min{(t_i, t_j)}$ loss elements through both modes of $\mathcal{B}$. The new network consisting of $\tilde{\M}_k$ and $\mathcal{B}$ has three properties:
\begin{itemize}
	\item[(i)] the shortest path starting from either input $i$ or $j$ has length $\min{(t_i, t_j)} + 1$, 
	\item[(ii)] there are $\min{(t_i, t_j)} + 1$ loss elements at both inputs modes $i$ and $j$, and
	\item[(iii)] the action of the network on modes other than $i$ and $j$ is unchanged.
\end{itemize}
Clearly the new network, $\tilde{\M}_{k+1}$, also satisfies the statement of our theorem. Therefore, by induction, this holds all the way until the network $\N$ is complete. Furthermore, this procedure is constructive by also giving explicitly the form of $\tilde{\N}$. \daniel{We conclude the proof by noting that the above scheme is also efficient, in the sense of requiring only polynomial effort in the number of modes and beam splitters in the network.}

\end{proof}

\begin{remark} As stated above, our result can be easily extended to include losses in the preparation and measurement stages of the experiment. Indeed, let $\Lambda_{\vec{\tran}_{in}}$ and $\Lambda_{\vec{\tran}_{out}}$ be channels describing (possibly non-uniform) losses occurring during the state-preparation and readout phases respectively. The quantum channel describing the whole physical process is then given by $ \Lambda_{tot}= \Lambda_{\vec{\tran}_{out}} \circ \Lambda_{\N} \circ \Lambda_{\vec{\tran}_{in}} $. Combining this with \eqref{eq:nonBALdec} we get a decomposition
\begin{equation}
\Lambda_{tot} = \Lambda_{\tilde{\N}'}\circ \Lambda_{\vec{\tran}'}\ ,
\end{equation}
where $ \Lambda_{\tilde{\N}'} = \Lambda_{\vec{\tran}_{out}} \Lambda_{\tilde{\N}} $ and $\vec{\tran}'= (\tran_{in,1}\tran_1,\tran_{in,2}\tran_2, \ldots , \tran_{in,m} \tran_m)$. 

\end{remark}

\daniel{
\begin{remark}
Inspection of the proof of \cref{thm:nonuniform} shows that the decomposition of Eq.\ \eqref{eq:nonBALdec} also holds in a network in which different optical elements induce different loss probabilities. In that case, it is still possible to write a decomposition similar to Eq.\ \eqref{eq:nonBALdec} where, for each input $i$, we pull out a factor 
\begin{equation*}
\eta_i= \max_{L_i} \prod_{j \in L_i} \eta_{i,j},
\end{equation*}
where $L_i$ are all paths connecting input $i$ to some output, $j$ indexes all linear-optical elements connected to input $i$, each with transmissivity $\eta_{i,j}$, and the product is taken over all $j$'s in path $L_i$. More explicitly, we pull out sufficient losses such that the final network $\tilde{\N}$ has one completely lossless path from input $i$ to some output. To see how this follows from the proof of the Theorem, note two things about the procedure of \cref{fig:iteration}: (i) it does not require each beam splitter to have equal loss, and (ii) at every iterative step, it leaves one output port of $\mathcal{B}$ lossless. Therefore, it will create one completely lossless path from each input to some output bu construction.

Finally we point out that, if each beam splitter has a different amount of loss, this might affect conclusions about finite-sized experiments drawn from our results, but it will not change any of our asymptotic conclusions.
\end{remark}}

\section{Simulation of boson sampling with binned inputs} \label{sec:result2}

In this Section we describe two new types of input states for which boson sampling is easy to simulate classically. One of them (type B) will be used later in a classical simulation for optical networks with nonuniform losses. {We start by providing a generalization of the algorithm by P.\ Clifford and R.\ Clifford \cite{Clifford2017}. This generalization, contrary to its predecessor, works for bunched input states, and similarly avoids an exponential cost in the number of modes $m$.  The proof of the following Theorem follows a similar reasoning to that of \cite{Clifford2017}. The main difference lays in the extensive usage of the first quantization formalism, which makes it easier to perform the necessary combinatorics and, in our opinion, can make the algorithm more approachable for physicists.  
}

\begin{theorem}[Generalization of the algorithm of {P.\ Clifford and R.\ Clifford} to collision input states] \label{thm:binnedBS}
Consider an $n$-photon, $m$-mode boson sampling instance described by linear optical transformation $U$ and an input state $\ket{\vec{S}}=\ket{s_1,s_2,\ldots,s_m}$. We do not assume $s_i\leq 1$. \daniel{There exists} a classical algorithm that generates a sample from the output probability distribution $p_U\left(\vec{S}\rightarrow\vec{T}\right)$ in time upper bounded by
  \begin{equation}\label{eq:timeOVERALL}
\mathcal{T} = nm \prod_{j=1}^m(s_j+1)^2  \alpha_S\ ,
\end{equation}
where $\alpha_S$ is the number of nonzero $s_i$'s. Moreover, the algorithm also outputs the probability corresponding to the sampled outcome.
\end{theorem}

\begin{proof}

Let $\vec{S} = ( s_1, s_2, \ldots s_m) $ encode the input state with $s_i$ photons in mode $i$ such that $|\vec{S}|=\sum s_i = n$. Define similar quantities for the output state $\vec{T} = (t_1, t_2 \ldots t_m)$. In what follows the initial configuration of photons $\vec{S}$ is fixed while $\vec{T}$ denote the possible outcomes of boson sampling experiment.

We split the proof into two parts. First, we recall the results of \cite{Chin2017} to give bounds on the computation cost on individual  probability amplitudes $p_U\left(\vec{S}\rightarrow\vec{T}\right)$. Then we show, by generalizing the ideas from \cite{Clifford2017}, how to sample from the output probability distribution by sampling the evolution of the individual particles one at the time, sidestepping the need to compute exponentially many individual probabilities. 

We start by reviewing results of \cite{Chin2017}. According to Eq.\ \eqref{eq:statistics}, the probability of some outcome $\vec{T}$ is proportional to $\lvert\Per(U_{ST})\lvert^2$. This permanent can be expanded as follows:
\begin{equation} \label{eq:permeff}
\Per(U_{ST}) = \frac{1}{2^n} \sum_{v_1=0}^{s_1} \cdots \sum_{v_m=0}^{s_m} (-1)^{N_v} \binom{s_1}{v_1} \cdots \binom{s_m}{v_m} \prod_{i=1}^m\left[\sum_{j=1}^m (s_j-2 v_j)U_{jk}\right]^{t_i},
\end{equation}
where $N_v=\sum_{i=1}^m v_i$. Let us compute the cost of calculating the permanent in this manner. First note there are $(s_i+1)$ terms for each of the $v_i$ sums. For each term in that multiple sum we need to compute
\begin{equation*}
\prod_{i=1}^m\left[\sum_{j=1}^m (s_j-2 v_j)U_{jk}\right]^{t_i}.
\end{equation*}
The product requires computing $\alpha_T$ terms (since those with $t_k = 0$ do not contribute), and the sum has $\alpha_S$ terms for similar reasons (i.e., if $s_i=0$ then $v_i=0$ as well). Since the absolute value of the  permanent is invariant under exchanging $S$ and $T$, computing $\Per(U_{ST})$ according to the expansion in \cref{eq:permeff} has a runtime of
\begin{equation}\label{eq:uppBOUND}
\tau_{S,T} =  \min\left(\prod_{i=1}^m(s_i+1),\prod_{j=1}^m(t_j+1) \right) \alpha_S \alpha_T \ .
\end{equation}

For the second part of the proof, we move from the usual notation of second quantization (or mode occupation lists) where a state is represented as a list of occupation numbers [i.e. $\vec{T} = ( t_1, t_2, \ldots t_m)$], to a notation native to first quantization  (or mode assignment lists). In this notation states are represented by multisets of $n$ elements in $[m]$, which we represent as tuples written in nondecreasing order, i.e.\ $\vec{z} = (z_1 \ldots z_n)$. Each $z_i$ represents in which mode photons $i$ is. Mapping back to the previous notation is easy: just identify $t_i$ as the multiplicity of mode $i$ in $\vec{z}$.  More details on the mapping between these descriptions can be found in  \cref{app:margprobs}.

The first step in constructing the sampling algorithm is to enlarge the sample space. Instead of sampling output states as multisets $\vec{z}$, we can sample a tuple $\vec{r} \in [m]^n$ and rewrite it in nondecreasing order to obtain the output $\vec{z}$. It is easy to show [see Eq.\eqref{eq:perPDF1stb} in \cref{app:margprobs}]  that sampling from the distribution
\begin{equation}\label{eq:perPDF}
p_U(\vec{S}\rightarrow\vec{r}) = \frac{1}{n!} \frac{1}{\prod_{i=1}^m s_i!} \left \lvert \Per{U_{S,\vec{r}}} \right \lvert^2,
\end{equation}
and then reordering $\vec{r}$ is equivalent to sampling from the correct boson sampling distribution. Notice that, in the above equation, we have expanded the sample space for the outputs but not for the inputs, since those are fixed (and of type A or B). We have also defined $U_{S,\vec{r}}$ in a similar spirit to $U_{S,T}$ just above \cref{eq:permeff}, but now for each $r_i$ in $\vec{r}$ we take row $r_i$ of $U_S$. The claim that sampling from \cref{eq:perPDF} is equivalent to sampling from the correct distribution follows from the fact that the permanent is invariant under permutations of rows and columns, and so it is equal for all $\vec{r}$'s that, when reordered, lead to the same $\vec{z}$.

The second step is to find the marginal probability mass function (pmf) of leading subsequences of $(r_1, r_2, \ldots r_n)$. The combinatorial factors that appear in the expression for this pmf make it cumbersome to write explicitly. In \cref{app:margprobs} we show how to compute these marginal pmfs, and show that they have two properties crucial for the simulation algorithm (which they share with the corresponding pmfs for no-collision inputs, i.e.\ Lemma 1 in \cite{Clifford2017}): 

\begin{itemize}
  \item[(i)] the probability of outcome $(r_1, r_2, \ldots r_l)$, for $l \leq n$, is the sum over probabilities corresponding to all inputs consistent with choosing $l$ photons from the actual $n$-photon input state $S$, and 
  \item[(ii)] each term in that sum corresponds to an input with the same binning structure as $S$. 
\end{itemize}
Therefore, if for some input $S$ the permanents $\Per(U_{ST})$ can be efficiently computed, then the same is true for \emph{individual marginal probabilities} of the distribution $ p_U(\vec{S}\rightarrow\vec{r})$.

We now describe the sampling algorithm. This is done by sampling the marginal distributions one at a time. That is, we use the chain rule for conditional probabilities to write
 \begin{equation} \label{eq:chain}
 p(\vec{r}) = p(r_1)p(r_2|r_1)p(r_3| r_1, r_2) \ldots p(r_n| r_1, \ldots, r_{n-1}).
 \end{equation}
It is easy to see that, for fixed $r_1$, sampling $r_2$ according to $p(r_2|r_1)$ is equivalent to sampling it according to (unnormalized) weights $p(r_1, r_2)$, and similarly for the subsequent $r_i$'s. Therefore, the above expression allows us to sample the elements of $(r_1, r_2, \ldots r_n)$ in order, one at a time, by computing the corresponding marginal probabilities. 

All that remains is to determine the running time necessary to generate a sample with the above procedure. There are $n$ steps. In step $l$, there are $m$ probabilities $p(r_1, \ldots, r_l)$ we need to compute according to the possible values of $r_l$ (recall that $r_1 $ to $ r_{l-1}$ are fixed by previous steps). Each probability requires computing some number of permanents, as per property (ii) above [cf.\ the discussion after \cref{eq:perPDFpartfinal} in \cref{app:margprobs}]. Each of these permanents can be computed in time equal to at most
\begin{equation}\label{eq:uppGLOB}
\tau= n \prod_{i=1}^m(s_i+1)\alpha_S  \ . 
\end{equation}
This follows from Eq.\ \eqref{eq:uppBOUND}, and details can be found in \cref{app:CompCost}. Therefore, the runtime of the algorithm can be upper bounded by 
\[
m\cdot\sum_{l=1}^n N_l(\vec{S}) \cdot \tau \ ,
\]
where $N_l(\vec{S})$ is the number of $l$-photon configurations \michal{ that can result by removing $n-l$ photons from  the initial configuration $\vec{S}$ [see discussion preceding Eq.\eqref{eq:comp}]}. Finally, note that $\sum_{l=1}^n N_l(\vec{S})= \prod_{i=1}^m(s_i+1) -1$, which concludes the proof.
\end{proof}

 {A straightforward inspection of Eq.\ \eqref{eq:timeOVERALL} shows that, for the standard collision-free boson sampling input (i.e.\ $s_i=1,\ i=1,\ldots,n$), our algorithm remains with an exponential (in $n$) running time.  However, using this formula we can identify two new families of input states for which boson sampling can be efficiently simulated classically.} 

\begin{corollary}[Efficient classical simulation of boson sampling for certain classes of unbalanced inputs] \label{cor:TwoTypesClass}
  Consider the following two types of input states:
  \begin{itemize}
    \item[A]: A Fock state with photons distributed arbitrarily in $k=$O$(1)$ different input modes (or bins).
    \item[B]: A Fock state with photons distributed in $k$=$ c\log n $ bins such that $k-1$ modes contain a single photon and all remaining photons are in the $k$th mode.
  \end{itemize}
Then, from \eqref{eq:timeOVERALL} it follows that there exists an efficient classical algorithm to simulate the corresponding instances of boson sampling. The runtimes of this algorithm for inputs of type A and B are, respectively,
\begin{align*}
\mathcal{T}_A & = k m n  (n+1)^{2k},  \\
\mathcal{T}_B&= c m n^{2c+3}  \log n.
\end{align*}
If $c$ and $k$ are constants, and in the standard boson sampling assumption that $m = $O$(n^2)$, we have
\begin{align*}
\mathcal{T}_A & = \textrm{O} \left(n^{3+2k} \right),  \\
\mathcal{T}_B&= \textrm{O} \left(n^{6+2c}\right).
\end{align*}
\end{corollary}

Note that $\mathcal{T}_B$ is not a consequence only of having $\O (\log n)$ bins. To see this, suppose we have $\log n$ bins with $n/\log n$ photons each. In this case the running time for the computation of a single permanent according to  Eq.\ \eqref{eq:uppGLOB} is
\begin{equation}
\tau = \O \left[\left(\frac{n}{\log n}\right)^{\log n} n \log n\right],
\end{equation}
which is super-polynomial in $n$. Therefore, the restriction that inputs of type B have most photons starting in a single mode is important.

We finally remark that types A and B of inputs states do not exhaust all the cases for which the reasoning given in the proof of Theorem \ref{thm:binnedBS} gives an efficient classical sampling algorithm. In particular, it was pointed out to us by Alex Moylett that the following class of input states, interpolating between types A and B,  is also classically simulable:  
  \begin{itemize}
    \item[C]: An $n$-photon Fock state with photons distributed \daniel{in} $k=k_A + k_B$ input modes, where  $k_A=\O(1)$ and $k_B = c \log n$. Moreover, photons located in modes ``A'' are distributed arbitrarily and remaining photons follow the distribution specified in the definition of input of type B.
    \end{itemize}

\section{Efficient simulation of lossy boson sampling with some lossless inputs} \label{sec:result3}

We now combine the ideas underlying the classical simulation of lossy boson sampling from \cite{Oszmaniec2018} with the results from the preceding section to expand the known scenarios for which boson sampling becomes classically simulable. We begin by recalling the notion of total variation distance---the figure of merit typically used to assess the accuracy of boson sampling \cite{Aaronson2013a}. The total variation distance between probability distributions $\lbrace p_x \rbrace$ and $\lbrace q_x \rbrace$  is given by
\begin{equation}\label{eq:TV}
d_\mathrm{TV}\left( \lbrace p_x \rbrace,\lbrace q_x \rbrace \right) = \frac{1}{2} \sum_x |p_x -q_x| \ .
\end{equation}
Importantly, the total variation distance between probabilities $\lbrace p^{\rho}_x \rbrace$ and $\lbrace p^{\sigma}_x \rbrace$ generated by measuring two quantum states $\rho$ and $\sigma$, respectively, with the same measurement $\mathbf{M}=\lbrace M_x \rbrace$ is always upper bounded by the trace distance between $\rho$ and $\sigma$,
\begin{equation}\label{eq:TVbound}
d_\mathrm{TV}\left( \lbrace p^{\rho}_x \rbrace, \lbrace p^{\sigma}_x \rbrace  \right) \leq d_\mathrm{tr}(\rho,\sigma)\ .  
\end{equation}

We are now ready to formally state our result.

\begin{theorem}[Classical simulation of lossy states when photons in some modes are not lost]\label{thm:notLOSTclass}
Consider an $m$-mode $n$-photon input state 
\begin{equation}
\ket{\Psi_0}=\ket{\overbrace{1\cdots 1}^n \overbrace{0\cdots 0}^{m-n}}
\end{equation} 
and a linear-optical transformation $\Lambda= \Lambda'\circ \Lambda_{k,\eta}$, where $\Lambda'$ is some (perhaps lossy) linear optical transformation, and $\Lambda_{k,\eta}$ is a pure loss channel such that the first $k$ modes have zero loss, while the remaining $m-k$ modes have transmission probability $\eta$, i.e.\ we have
\begin{equation}\label{eq:ASSloss}
\vec{\eta}=(\overbrace{1,\ldots,1}^k, \overbrace{\eta,\ldots,\eta}^{m-k})\ .
\end{equation} 
Let $\lbrace p^{\Lambda}_x\rbrace$ be the probability distribution obtained by measuring the output state $ \Lambda \left(\ket{\Psi_0}\bra{\Psi_0} \right)$ with particle-number measurements. Then, provided $k=\O(\log n)$, for any $\Lambda$ \daniel{there exists} a (weakly) classically simulable distribution $\lbrace q^{\Lambda}_x \rbrace$ which satisfies
\begin{equation}
d_\mathrm{TV}\left( \lbrace p^{\Lambda}_x \rbrace,\lbrace q^{\Lambda}_x \rbrace \right) \leq   \frac{\eta^2 (n-k)}{2} +\frac{\eta(1-\eta)}{2}\ .
\end{equation}
Consequently, if \daniel{$\eta = o(1/\sqrt{n-k})$}, the error incurred in approximating $ \lbrace p^{\Lambda}_x \rbrace$ with the (classically easy) distribution $\lbrace q^{\Lambda}_x \rbrace$ goes to $0$ as the number of particles $n$ tends to infinity.
\end{theorem}

\begin{proof}[Proof sketch]
The proof generalizes the main idea from \cite{Oszmaniec2018} where the classical simulation of lossy linear optics was based on approximating the lossy bosonic state on $\tilde{n}$ particles
 by particle-separable states, i.e.\ states of the form
\begin{equation}\label{eq:partSEP}
\tilde{\sigma}= \sum_\alpha \tilde{P}_\alpha  \ketbra{\phi_\alpha}{\phi_\alpha}^{\otimes n_\alpha}\ ,
\end{equation}
Importantly, state $\ketbra{\phi_\alpha}{\phi_\alpha}^{\otimes n_\alpha}$ can be generated from $\ket{n_{\alpha} 0\cdots 0}$ by application of linear-optical transformation $\Lambda_{U_\alpha}$.

The lossy state considered in \cite{Oszmaniec2018} was of the form $\tilde{\rho}=\Lambda_\eta (\ketbra{\Phi_0}{\Phi_0})$, where $\Lambda_\eta$ is the channel that describes uniform losses [see Eq.\eqref{eq:BSuniformMODEL}] and $ \ketbra{\Phi_0}{\Phi_0}$ is a standard boson sampling input state on $\tilde{n}$ particles.  It was shown that \daniel{there exists a} particle-separable state \michal{ $\tilde{\sigma}$ and a} probability distribution $\lbrace \tilde{P}_\alpha\rbrace$ that can be sampled efficiently, and such that the trace distance to the target state $\tilde{\rho}$ is upper bounded by 
\begin{equation}\label{eq:appPRV}
d_{\mathrm{tr}}(\tilde{\rho},\tilde{\sigma})\leq \tilde{\Delta} = \frac{\eta^2 \tilde{n}}{2} +\frac{\eta(1-\eta)}{2}\ .
\end{equation}
From the above and the fact that boson sampling with input states $\ketbra{\psi_\alpha}{\psi_\alpha}^{\otimes n_\alpha}$ can be simulated classically it follows that any linear-optical process $\tilde{\Lambda}$ applied to $\tilde{\rho}$ and terminating with particle-number measurements can also be (weakly) simulated classically within total variation distance $\tilde{\Delta}$.

Here we generalize the above argument to the situation at hand. The key difference is that, as the (partially) lossy input state, we take the state $\rho = \Lambda_{k,\eta}(\ketbra{\Psi_0}{\Psi_0})$ for which  no particles are lost  in the first $k$ modes. We can approximate $\rho$ by
\begin{equation}
\sigma =\sum_\alpha P_\alpha \ketbra{\psi_\alpha}{\psi_\alpha}\ ,
\end{equation}    
where, in contrast to Eq.\ \eqref{eq:partSEP}, states entering in the above convex decomposition are of the type $\Lambda_{U_\alpha} (\ketbra{\Psi_{k,\alpha}}{\Psi_{k,\alpha}})$, where $\ket{\Psi_{k,\alpha}}=\ket{\overbrace{1\cdots 1}^k,n_{\alpha},0\cdots 0}$  and $\Lambda_{U_\alpha}$ is a suitable linear optical unitary transformation that acts as the identity on the first $k$ optical modes. If $k=\O(\log n )$, states of type $\ket{\Psi_{k,\alpha}}$ lead to classically easy boson sampling instances (they belong to type B appearing in Theorem \ref{thm:binnedBS}). Repeating the considerations that led to Eq.\eqref{eq:appPRV} we conclude it is possible to choose a distribution $\lbrace P_\alpha \rbrace$ which can be sampled efficiently and for which we have
\begin{equation}
d_{\mathrm{tr}}(\rho,\sigma)\leq \Delta = \frac{\eta^2 (n-k)}{2} +\frac{\eta(1-\eta)}{2}\ .
\end{equation}  
By the assumptions stated above, sampling from $\Lambda'(\sigma)$ can be done efficiently, which concludes the proof. 
\end{proof}

\section{Classical simulation of linear optics in unbalanced lossy networks} \label{sec:result4}

In this Section we combine results from Theorems \ref{thm:nonuniform} and \ref{thm:notLOSTclass} to show that, for arbitrary lossy linear networks, boson sampling becomes classically simulable provided there are at most O($\log n$) input modes for which the input-output paths are short.

\begin{theorem}[Classical simulation in unbalanced lossy networks]\label{thm:notUNIFclass}
Consider an $m$-mode lossy optical linear network $\N$ acting on the standard input $n$-photon state of the form
\begin{equation}
\ket{\Psi_0}=\ket{\overbrace{1\cdots 1}^n \overbrace{0\cdots 0}^{m-n}}\ .
\end{equation}
Assume every beam splitter introduces losses, characterized by transmissivity $\eta$, on the modes on which it acts. Moreover, assume that at most $\O( \log n)$ \daniel{inputs} are such that their shortest input-output path has length smaller than $c \log n$ for some suitable constant $c>0$. Then \daniel{there exists} a classical algorithm to sample from a $\lbrace q^{\N}_x \rbrace$ that approximates the target  output distribution $\lbrace p^{\N}_x \rbrace$, obtained by evolving the input state through network $\N$, to accuracy
\begin{equation}\label{eq:AAPRX}
d_\mathrm{TV}\left( \lbrace p^{\N}_x \rbrace,\lbrace q^{\N}_x \rbrace \right)=\Delta_\N \leq \frac{\eta_\mathrm{eff}^2 (n-k)}{2} +\frac{\eta_\mathrm{eff}(1-\eta_\mathrm{eff})}{2}
\end{equation}
where $\eta_{\mathrm{eff}}=\eta^{c \log(n)}= n^{- c \log(1/\eta)}$. In particular, if $c>\frac{1}{2\log(1/\eta)}$, then $\Delta_\N \rightarrow 0$ as the number of photons $n$ increase.
\end{theorem}

\begin{proof}
The result follows directly from Theorems \ref{thm:nonuniform} and \ref{thm:notLOSTclass}. First, note that at most $k$ input modes have their shortest input-output path of length $s_i \leq c \log n$. Then, by Theorem \ref{thm:nonuniform},  we can extract from the network a layer of nonuniform losses that, up to relabeling of modes, is of the form
\begin{equation}\label{eq:simpLOSS}
\vec{\eta}= (\overbrace{1,\ldots,1}^k, \overbrace{\eta_\mathrm{eff},\ldots,\eta_\mathrm{eff}}^{m-k})\ .
\end{equation}
Next, note that assumptions of Theorem \ref{thm:notLOSTclass} are met since $k=\O(\log n)$ and the vector of losses from Eq.\ \eqref{eq:simpLOSS} directly corresponds to  Eq.\ \eqref{eq:ASSloss}. By directly invoking of Theorem \ref{thm:notLOSTclass} we obtain Eq.\ \eqref{eq:AAPRX}, which concludes the proof. 
\end{proof}

\begin{exa}[Triangular decomposition of linear-optical networks \cite{Reck1994}]
Consider the standard construction of linear-optical networks given in Fig.\ \ref{fig:losspullout}. Assume that every beam splitter in the network is lossy. The number of beam-splitters photons need to cross, as they propagate from input $i$ (counting from the bottom of the network) scales linearly \daniel{with} $i$. Therefore, a lossy network in this geometry satisfies the assumptions of Theorem \ref{thm:notUNIFclass} and the corresponding boson sampling instance can be simulated as described in the statement of the Theorem.
\end{exa}

\daniel{The conclusions of \cref{thm:notUNIFclass} are of an asymptotic nature. We note that the error incurred in the classical simulation described here is the same as in \cite{Oszmaniec2018}, and we direct the interested reader to that reference for a discussion on how this error scales for finite-sized experiments.}

\section{Conclusions and outlook}

We presented a comprehensive study of the effect of nonuniform losses on classical simulability of boson sampling. Our main conclusion is that linear-optical networks for which losses exhibit high degree of non-uniformity cannot evade a classical simulation strategies based on ideas from \cite{GarciaPatron2017,Oszmaniec2018}. Our main contributions are twofold. First, in \cref{thm:nonuniform} we established the quantitative relation between the geometric properties of lossy networks and the amount of losses that can be \emph{pulled-out} to the input of the network. Second, in \cref{thm:notUNIFclass} we provided sufficient conditions for effective classical simulability of of boson sampling subject to nonuniform losses (that can result, for example, from the complicated network structure). To establish the second result we extended the  classical simulation algorithm from \cite{Clifford2017} to instances of boson sampling with binned inputs. This allowed us to identify two classes of input states for which the simulation of boson sampling becomes efficient: (A)  when $n$ input photons occupy a constant number of input modes; (B) when all but $O(\log n)$ photons are concentrated on a single input mode, while an additional $O(\log n)$ modes contain one photon each.

We believe that our results can help in understanding which linear-optical designs can offer scalable platforms for demonstration of quantum advantage based on boson sampling. Moreover, the classical simulation techniques presented here can be leveraged to benchmark the performance of photonic networks with complicated input states in the near future. We would like to conclude by stating a number of interesting open questions. 

First, in a recent work \cite{Renema2018} Renema \emph{et al.} proved that, for \emph{typical} (i.e.\ Haar-random) linear-optical interferometers $U$, boson sampling can be classically simulated if only a \emph{constant fraction} of input photons are lost and every optical mode suffers from identical losses.  The technique used in \cite{Renema2018} was based on approximating  the lossy output probability distributions by distributions originating from interference of partially distinguishable photons (with  increasing degree of indistinguishability). It is an interesting question whether input states of type A or type B can be used to obtain similar results for generic multi-mode interferometers.

Another interesting question is what are the consequences of our result for Gaussian boson sampling \cite{GaussBS2016}, a variant of boson sampling that uses squeezed Gaussian states as input. A recent work \cite{Brod2019} studied under which conditions lossy Gaussian boson sampling becomes easy to simulate. We conjecture that our results from \cref{sec:result2} can be adapted to the case of Gaussian boson sampling to prove, for example, that it is also classically simulable if the squeezed states are inputs only in $k=$\textrm{O}$(1)$. We also conjecture that our results of \cref{sec:result4} can be adapted to show that the conclusions of \cite{Brod2019} also hold for unbalanced networks (there, the authors made the standard assumption of uniform losses).

Lastly, nonuniform losses clearly seem to undermine the quantum features of nonclassical states. It is natural to explore the consequences of our results regarding extraction of nonuniform losses from a lossy network (\cref{thm:nonuniform}) for the nonclassical features of quantum states generated by such networks \cite{MultidimOptics2018}.
 
\begin{acknowledgments}
We thank Micha\l\ Horodecki for interesting and fruitful discussions. 
M.\ O.\ acknowledges the financial support by TEAM-NET project (contract no. POIR.04.04.00-00-17C1/18-00). D.\ J.\ B.\ acknowledges funding from Instituto Nacional de Ci\^encia e Tecnologia de Informa\c{c}\~ao Qu\^antica (INCT-IQ/CNPq).
\end{acknowledgments}

\bibliographystyle{unsrtnat}
\bibliography{bosonsamplingrefs2.bib}

\begin{thebibliography}{45}
\providecommand{\natexlab}[1]{#1}
\providecommand{\url}[1]{\texttt{#1}}
\expandafter\ifx\csname urlstyle\endcsname\relax
  \providecommand{\doi}[1]{doi: #1}\else
  \providecommand{\doi}{doi: \begingroup \urlstyle{rm}\Url}\fi

\bibitem[Harrow and Montanaro(2017)]{Harrow2017}
A.~Harrow and A.~Montanaro.
\newblock Quantum computational supremacy.
\newblock \emph{Nature}, 549:\penalty0 203--209, 2017.
\newblock \doi{10.1038/nature23458}.

\bibitem[Aaronson and Arkhipov(2013)]{Aaronson2013a}
S.~Aaronson and A.~Arkhipov.
\newblock The computational complexity of linear optics.
\newblock \emph{Theory of Computing}, 4:\penalty0 143--252, 2013.
\newblock \doi{10.4086/toc.2013.v009a004}.

\bibitem[Brod et~al.(2019)Brod, Galvão, Crespi, Osellame, Spagnolo, and
  Sciarrino]{RevPhot2019}
D.~J. Brod, E.~F. Galvão, A.~Crespi, R.~Osellame, N.~Spagnolo, and
  F.~Sciarrino.
\newblock Photonic implementation of boson sampling: a review.
\newblock \emph{Advanced Photonics}, 1\penalty0 (3):\penalty0 1--14, 2019.
\newblock \doi{10.1117/1.AP.1.3.034001}.

\bibitem[Broome et~al.(2013)Broome, Fedrizzi, Rahimi-Keshari, Dove, Aaronson,
  Ralph, and White]{Broome2013}
M.~A. Broome, A.~Fedrizzi, S.~Rahimi-Keshari, J.~Dove, S.~Aaronson, T.~C.
  Ralph, and A.~G. White.
\newblock {Photonic Boson Sampling} in a tunable circuit.
\newblock \emph{Science}, 339\penalty0 (6121):\penalty0 794, 2013.
\newblock \doi{10.1126/science.1231440}.

\bibitem[Crespi et~al.(2013)Crespi, Osellame, Ramponi, Brod, Galv{\~a}o,
  Spagnolo, Vitelli, Maiorino, Mataloni, and Sciarrino]{Crespi2013b}
A.~Crespi, R.~Osellame, R.~Ramponi, D.~J. Brod, E.~F. Galv{\~a}o, N.~Spagnolo,
  C.~Vitelli, E.~Maiorino, P.~Mataloni, and F.~Sciarrino.
\newblock Integrated multimode interferometers with arbitrary designs for
  photonic {Boson Sampling}.
\newblock \emph{Nat. Photon.}, 7\penalty0 (7):\penalty0 545--549, 2013.
\newblock \doi{10.1038/nphoton.2013.112}.

\bibitem[Spring et~al.(2013)Spring, Metcalf, Humphreys, Kolthammer, Jin,
  Barbieri, Datta, {Thomas-Peter}, Langford, Kundys, Gates, Smith, Smith, and
  Walmsley]{Spring2013}
J.~B. Spring, B.~J. Metcalf, P.~C. Humphreys, W.~S. Kolthammer, {X.-M.} Jin,
  M.~Barbieri, A.~Datta, N.~{Thomas-Peter}, N.~K. Langford, D.~Kundys, J.~C.
  Gates, B.~J. Smith, P.~G.~R. Smith, and I.~A. Walmsley.
\newblock Boson {{Sampling}} on a {{Photonic Chip}}.
\newblock \emph{Science}, 339\penalty0 (6121):\penalty0 798, 2013.
\newblock \doi{10.1126/science.1231692}.

\bibitem[Tillmann et~al.(2013)Tillmann, Daki{\'c}, Heilmann, Nolte, Szameit,
  and Walther]{Tillmann2013}
M.~Tillmann, B.~Daki{\'c}, R.~Heilmann, S.~Nolte, A.~Szameit, and P.~Walther.
\newblock Experimental {Boson Sampling}.
\newblock \emph{Nat. Photon.}, 7\penalty0 (7):\penalty0 540--544, 2013.
\newblock \doi{10.1038/nphoton.2013.102}.

\bibitem[Spagnolo et~al.(2014)Spagnolo, Vitelli, Bentivegna, Brod, Crespi,
  Flamini, Giacomini, Milani, Ramponi, Mataloni, Osellame, Galv{\~a}o, and
  Sciarrino]{Spagnolo2014}
N.~Spagnolo, C.~Vitelli, M.~Bentivegna, D.~J. Brod, A.~Crespi, F.~Flamini,
  S.~Giacomini, G.~Milani, R.~Ramponi, P.~Mataloni, R.~Osellame, E.~F.
  Galv{\~a}o, and F.~Sciarrino.
\newblock Experimental validation of photonic {Boson Sampling}.
\newblock \emph{Nat. Photon.}, 8\penalty0 (8):\penalty0 615--620, 2014.
\newblock \doi{10.1038/nphoton.2014.135}.

\bibitem[Carolan et~al.(2014)Carolan, Meinecke, Shadbolt, Russell, Ismail,
  W{\"o}rhoff, Rudolph, Thompson, O'Brien, Matthews, and Laing]{Carolan2014}
J.~Carolan, J.~D.~A. Meinecke, P.~J. Shadbolt, N.~J. Russell, N.~Ismail,
  K.~W{\"o}rhoff, T.~Rudolph, M.~G. Thompson, J.~L. O'Brien, J.~C.~F. Matthews,
  and A.~Laing.
\newblock On the experimental verification of quantum complexity in linear
  optics.
\newblock \emph{Nat. Photon.}, 8\penalty0 (8):\penalty0 621--626, 2014.
\newblock \doi{10.1038/nphoton.2014.152}.

\bibitem[Carolan et~al.(2015)Carolan, Harrold, Sparrow, Mart{\'\i}n-L{\'o}pez,
  Russell, Silverstone, Shadbolt, Matsuda, Oguma, Itoh, Marshall, Thompson,
  Matthews, Hashimoto, O'Brien, and Laing]{Carolan2015}
J.~Carolan, C.~Harrold, C.~Sparrow, E.~Mart{\'\i}n-L{\'o}pez, N.~J. Russell,
  J.~W. Silverstone, P.~J. Shadbolt, N.~Matsuda, M.~Oguma, M.~Itoh, G.~D.
  Marshall, M.~G. Thompson, J.~C.~F. Matthews, T.~Hashimoto, J.~L. O'Brien, and
  A.~Laing.
\newblock Universal linear optics.
\newblock \emph{Science}, 349\penalty0 (6249):\penalty0 711, 2015.
\newblock \doi{10.1126/science.aab3642}.

\bibitem[Bentivegna et~al.(2015)Bentivegna, Spagnolo, Vitelli, Flamini,
  Viggianiello, Latmiral, Mataloni, Brod, Galv\~ao, Crespi, Ramponi, Osellame,
  and Sciarrino]{Bentivegna2015}
M.~Bentivegna, N.~Spagnolo, C.~Vitelli, F.~Flamini, N.~Viggianiello,
  L.~Latmiral, P.~Mataloni, D.~J. Brod, E.~F. Galv\~ao, A.~Crespi, R.~Ramponi,
  R.~Osellame, and F.~Sciarrino.
\newblock Experimental scattershot {Boson Sampling}.
\newblock \emph{Science Advances}, 1\penalty0 (3):\penalty0 e1400255, 2015.
\newblock \doi{10.1126/sciadv.1400255}.

\bibitem[Loredo et~al.(2017)Loredo, Broome, Hilaire, Gazzano, Sagnes, Lemaitre,
  Almeida, Senellart, and White]{Loredo2017}
J.~C. Loredo, M.~A. Broome, P.~Hilaire, O.~Gazzano, I.~Sagnes, A.~Lemaitre,
  M.~P. Almeida, P.~Senellart, and A.~G. White.
\newblock {Boson Sampling} with single-photon fock states from a bright
  solid-state source.
\newblock \emph{Phys. Rev. Lett.}, 118:\penalty0 130503, 2017.
\newblock \doi{10.1103/PhysRevLett.118.130503}.

\bibitem[He et~al.(2017)He, Ding, Su, Huang, Qin, Wang, Unsleber, Chen, Wang,
  He, Wang, Zhang, Chen, Schneider, Kamp, You, Wang, H\"ofling, Lu, and
  Pan]{He2017}
Y.~He, X.~Ding, Z.-E. Su, H.-L. Huang, J.~Qin, C.~Wang, S.~Unsleber, C.~Chen,
  H.~Wang, Y.-M. He, X.-L. Wang, W.-J. Zhang, S.-J. Chen, C.~Schneider,
  M.~Kamp, L.-X. You, Z.~Wang, S.~H\"ofling, Chao-Yang Lu, and Jian-Wei Pan.
\newblock Time-bin-encoded {Boson Sampling} with a single-photon device.
\newblock \emph{Phys. Rev. Lett.}, 118:\penalty0 190501, 2017.
\newblock \doi{10.1103/PhysRevLett.118.190501}.

\bibitem[Wang et~al.(2017)Wang, He, Li, Su, Li, Huang, Ding, Chen, Liu, Qin,
  Li, He, Schneider, Kamp, Peng, H{\"o}fling, Lu, and Pan]{Wang2017}
H.~Wang, Y.~He, {Y.-H.} Li, {Z.-E.} Su, B.~Li, {H.-L.} Huang, X.~Ding, {M.-C.}
  Chen, C.~Liu, J.~Qin, {J.-P.} Li, {Y.-M.} He, C.~Schneider, M.~Kamp, {C.-Z.}
  Peng, S.~H{\"o}fling, {C.-Y.} Lu, and {J.-W.} Pan.
\newblock High-efficiency multiphoton boson sampling.
\newblock \emph{Nat. Photon.}, 11:\penalty0 361--365, 2017.
\newblock \doi{10.1038/nphoton.2017.63}.

\bibitem[Wang et~al.(2018{\natexlab{a}})Wang, Li, Jiang, He, Li, Ding, Chen,
  Qin, Peng, Schneider, Kamp, Zhang, Li, You, Wang, Dowling, H\"ofling, Lu, and
  Pan]{Wang2018}
H.~Wang, W.~Li, X.~Jiang, Y.-M. He, Y.-H. Li, X.~Ding, M.-C. Chen, J.~Qin,
  C.-Z. Peng, C.~Schneider, M.~Kamp, W.-J. Zhang, H.~Li, L.-X. You, Z.~Wang,
  J.~P. Dowling, S.~H\"ofling, Chao-Yang Lu, and Jian-Wei Pan.
\newblock Toward scalable {Boson Sampling} with photon loss.
\newblock \emph{Phys. Rev. Lett.}, 120:\penalty0 230502, 2018{\natexlab{a}}.
\newblock \doi{10.1103/PhysRevLett.120.230502}.

\bibitem[Neville et~al.(2017)Neville, Sparrow, Clifford, Johnston, Birchall,
  Montanaro, and Laing]{Neville2017}
A.~Neville, C.~Sparrow, R.~Clifford, E.~Johnston, P.~M. Birchall, A.~Montanaro,
  and A.~Laing.
\newblock Classical boson sampling algorithms with superior performance to
  near-term experiments.
\newblock \emph{Nature Physics}, 13:\penalty0 1153, 2017.
\newblock \doi{doi:10.1038/nphys4270}.

\bibitem[Clifford and Clifford(2018)]{Clifford2017}
P.~Clifford and R.~Clifford.
\newblock The classical complexity of {Boson Sampling}.
\newblock In \emph{Proceedings of the Twenty-Ninth Annual ACM-SIAM Symposium on
  Discrete Algorithms}, page 146, 2018.
\newblock \doi{10.1137/1.9781611975031.10}.

\bibitem[{Roga} and {Takeoka}(2019)]{Roga2019}
W.~{Roga} and M.~{Takeoka}.
\newblock {Classical simulation of boson sampling with sparse output}.
\newblock 2019.
\newblock arXiv:1904.05494.

\bibitem[{Wu} et~al.(2018){Wu}, {Liu}, {Zhang}, {Jin}, {Wang}, {Wang}, and
  {Yang}]{Wu2016}
J.~{Wu}, Y.~{Liu}, B.~{Zhang}, X.~{Jin}, Y.~{Wang}, H.~{Wang}, and X.~{Yang}.
\newblock A benchmark test of boson sampling on {Tianhe-2} supercomputer.
\newblock \emph{National Science Review}, 5:\penalty0 715, 2018.
\newblock \doi{10.1093/nsr/nwy079}.

\bibitem[{Dalzell} et~al.(2018){Dalzell}, {Harrow}, {Koh}, and {La
  Placa}]{Dalzell2018}
A.~M. {Dalzell}, A.~W. {Harrow}, D.~E. {Koh}, and R.~L. {La Placa}.
\newblock {How many qubits are needed for quantum computational supremacy?}
\newblock 2018.
\newblock arXiv:1805.05224.

\bibitem[Aaronson and Brod(2016)]{Aaronson2016}
S.~Aaronson and D.~J. Brod.
\newblock {BosonSampling} with lost photons.
\newblock \emph{Phys. Rev. A}, 93\penalty0 (1):\penalty0 012335, 2016.
\newblock \doi{10.1103/PhysRevA.93.012335}.

\bibitem[{Garc{\'{\i}}a-Patr{\'o}n} et~al.(2019){Garc{\'{\i}}a-Patr{\'o}n},
  {Renema}, and {Shchesnovich}]{GarciaPatron2017}
R.~{Garc{\'{\i}}a-Patr{\'o}n}, J.~J. {Renema}, and V.~{Shchesnovich}.
\newblock Simulating boson sampling in lossy architectures.
\newblock \emph{Quantum}, 3:\penalty0 169, 2019.
\newblock \doi{10.22331/q-2019-08-05-169}.

\bibitem[Oszmaniec and Brod(2018)]{Oszmaniec2018}
M.~Oszmaniec and D.~Brod.
\newblock Classical simulation of photonic linear optics with lost particles.
\newblock \emph{New Journal of Physics}, 20\penalty0 (9):\penalty0 092002,
  2018.
\newblock \doi{10.1088/1367-2630/aadfa8}.

\bibitem[{Renema} et~al.(2018){Renema}, {Shchesnovich}, and
  Garc{\'{\i}}a-Patr{\'o}n]{Renema2018}
J.~{Renema}, V.~{Shchesnovich}, and R.~Garc{\'{\i}}a-Patr{\'o}n.
\newblock {Classical simulability of noisy boson sampling}.
\newblock 2018.
\newblock arXiv:1809.01953.

\bibitem[Moylett et~al.(2019)Moylett, Garc{\'{\i}}a-Patr{\'{o}}n, Renema, and
  Turner]{Moylett2019}
A.~E. Moylett, R.~Garc{\'{\i}}a-Patr{\'{o}}n, J.~Renema, and P.~Turner.
\newblock Classically simulating near-term partially-distinguishable and lossy
  boson sampling.
\newblock \emph{Quantum Science and Technology}, 5\penalty0 (1):\penalty0
  015001, 2019.
\newblock \doi{10.1088/2058-9565/ab5555}.

\bibitem[Arkhipov(2015)]{Arkhipov2015}
A.~Arkhipov.
\newblock {BosonSampling} is robust against small errors in the network matrix.
\newblock \emph{Phys. Rev. A}, 92:\penalty0 062326, 2015.
\newblock \doi{10.1103/PhysRevA.92.062326}.

\bibitem[Leverrier and {Garc{\'{\i}}a-Patr{\'o}n}(2014)]{Leverrier2014}
A.~Leverrier and R.~{Garc{\'{\i}}a-Patr{\'o}n}.
\newblock Analysis of circuit imperfections in {BosonSampling}.
\newblock \emph{Quant. Inf. Comp.}, 15:\penalty0 489, 2014.

\bibitem[Rahimi-Keshari et~al.(2016)Rahimi-Keshari, Ralph, and
  Caves]{Rahimi-Keshari2016}
S.~Rahimi-Keshari, T.~C. Ralph, and C.~M. Caves.
\newblock Sufficient conditions for efficient classical simulation of quantum
  optics.
\newblock \emph{Phys. Rev. X}, 6\penalty0 (2):\penalty0 021039, 2016.
\newblock \doi{10.1103/PhysRevX.6.021039}.

\bibitem[Renema et~al.(2018)Renema, Menssen, Clements, Triginer, Kolthammer,
  and Walmsley]{Renema2017}
J.~J. Renema, A.~Menssen, W.~R. Clements, G.~Triginer, W.~S. Kolthammer, and
  I.~A. Walmsley.
\newblock Efficient algorithm for boson sampling with partially distinguishable
  photons.
\newblock \emph{Phys. Rev. Lett.}, 120:\penalty0 220502, 2018.
\newblock \doi{10.1103/PhysRevLett.120.220502}.

\bibitem[{Kalai} and {Kindler}(2014)]{Kalai2014}
G.~{Kalai} and G.~{Kindler}.
\newblock {Gaussian Noise Sensitivity and {BosonSampling}}.
\newblock 2014.
\newblock arXiv:1409.3093.

\bibitem[Reck et~al.(1994)Reck, Zeilinger, Bernstein, and Bertani]{Reck1994}
M.~Reck, A.~Zeilinger, H.~Bernstein, and P.~Bertani.
\newblock Experimental realization of any discrete unitary operator.
\newblock \emph{Phys. Rev. Lett.}, 73:\penalty0 58, 1994.
\newblock \doi{10.1103/PhysRevLett.73.58}.

\bibitem[{Arute} et~al.(2019)]{Suprem2019}
F.~{Arute} et~al.
\newblock {Quantum supremacy using a programmable superconducting processor}.
\newblock \emph{Nature}, 574\penalty0 (7779):\penalty0 505--510, 2019.
\newblock \doi{10.1038/s41586-019-1666-5}.

\bibitem[Boixo et~al.(2018)Boixo, Isakov, Smelyanskiy, Babbush, Ding, Jiang,
  Bremner, Martinis, and Neven]{Boixo2016}
S.~Boixo, S.~Isakov, V.~Smelyanskiy, R.~Babbush, N.~Ding, Z.~Jiang, M.~Bremner,
  J.~Martinis, and H.~Neven.
\newblock Characterizing quantum supremacy in near-term devices.
\newblock \emph{Nature Physics}, 14\penalty0 (6):\penalty0 595--600, 2018.
\newblock \doi{10.1038/s41567-018-0124-x}.

\bibitem[Wang et~al.(2019)Wang, Qin, Ding, Chen, Chen, You, He, Jiang, You,
  Wang, Schneider, Renema, H\"ofling, Lu, and Pan]{BosNew2019}
H.~Wang, J.~Qin, X.~Ding, M.-C. Chen, S.~Chen, X.~You, Y.-M. He, X.~Jiang,
  L.~You, Z.~Wang, C.~Schneider, J.~Renema, S.~H\"ofling, C.-Y. Lu, and J.-W.
  Pan.
\newblock {Boson Sampling} with 20 input photons and a 60-mode interferometer
  in a $1{0}^{14}$-dimensional {Hilbert} space.
\newblock \emph{Phys. Rev. Lett.}, 123:\penalty0 250503, 2019.
\newblock \doi{10.1103/PhysRevLett.123.250503}.

\bibitem[Moylett and Turner(2018)]{Moylett2018}
A.~E. Moylett and P.~S. Turner.
\newblock Quantum simulation of partially distinguishable boson sampling.
\newblock \emph{Phys. Rev. A}, 97:\penalty0 062329, 2018.
\newblock \doi{10.1103/PhysRevA.97.062329}.

\bibitem[Chin and Huh(2018)]{Chin2017}
S.~Chin and J.~Huh.
\newblock Generalized concurrence in boson sampling.
\newblock \emph{Sci. Rep.}, 8:\penalty0 6101, 2018.
\newblock \doi{10.1038/s41598}.

\bibitem[Shchesnovich(2013)]{Valery2013}
V.~S. Shchesnovich.
\newblock Asymptotic evaluation of bosonic probability amplitudes in linear
  unitary networks in the case of large number of bosons.
\newblock \emph{International Journal of Quantum Information}, 11\penalty0
  (05):\penalty0 1350045, 2013.
\newblock \doi{10.1142/S0219749913500457}.

\bibitem[Bouland and Aaronson(2014)]{Bouland2014}
A.~Bouland and S.~Aaronson.
\newblock Generation of universal linear optics by any beam splitter.
\newblock \emph{Phys. Rev. A}, 89:\penalty0 062316, 2014.
\newblock \doi{10.1103/PhysRevA.89.062316}.

\bibitem[Sawicki(2016)]{Sawicki2016}
A.~Sawicki.
\newblock Universality of beamsplitters.
\newblock \emph{Quantum Inf. Comput.}, 16\penalty0 (3 and 4):\penalty0
  0291--0312, 2016.

\bibitem[{Sawicki} and {Karnas}(2017)]{SawickiKarnas2017}
A.~{Sawicki} and K.~{Karnas}.
\newblock {Universality of Single-Qudit Gates}.
\newblock \emph{Annales Henri Poincare}, 18\penalty0 (11):\penalty0 3515--3552,
  2017.
\newblock \doi{10.1007/s00023-017-0604-z}.

\bibitem[Tischler et~al.(2018)Tischler, Rockstuhl, and S\l{}owik]{Slowik2018}
N.~Tischler, C.~Rockstuhl, and K.~S\l{}owik.
\newblock Quantum optical realization of arbitrary linear transformations
  allowing for loss and gain.
\newblock \emph{Phys. Rev. X}, 8:\penalty0 021017, 2018.
\newblock \doi{10.1103/PhysRevX.8.021017}.

\bibitem[Clements et~al.(2016)Clements, Humphreys, Metcalf, Kolthammer, and
  Walmsley]{Clements2016}
W.~R. Clements, P.~C. Humphreys, B.~J. Metcalf, W.~S. Kolthammer, and I.~A.
  Walmsley.
\newblock Optimal design for universal multiport interferometers.
\newblock \emph{Optica}, 3\penalty0 (12):\penalty0 1460--1465, 2016.
\newblock \doi{10.1364/OPTICA.3.001460}.

\bibitem[Hamilton et~al.(2017)Hamilton, Kruse, Sansoni, Barkhofen, Silberhorn,
  and Jex]{GaussBS2016}
C.~S. Hamilton, R.~Kruse, L.~Sansoni, S.~Barkhofen, C.~Silberhorn, and I.~Jex.
\newblock Gaussian {Boson Sampling}.
\newblock \emph{Phys. Rev. Lett.}, 119:\penalty0 170501, 2017.
\newblock \doi{10.1103/PhysRevLett.119.170501}.

\bibitem[Qi et~al.(2020)Qi, Brod, Quesada, and Garc\'{\i}a-Patr\'on]{Brod2019}
H.~Qi, D.~J. Brod, N.~Quesada, and R.~Garc\'{\i}a-Patr\'on.
\newblock Regimes of classical simulability for noisy {Gaussian} boson
  sampling.
\newblock \emph{Phys. Rev. Lett.}, 124:\penalty0 100502, 2020.
\newblock \doi{10.1103/PhysRevLett.124.100502}.

\bibitem[Wang et~al.(2018{\natexlab{b}})Wang, Paesani, Ding, Santagati,
  Skrzypczyk, Salavrakos, Tura, Augusiak, Man{\v c}inska, Bacco, Bonneau,
  Silverstone, Gong, Ac{\'\i}n, Rottwitt, Oxenl{\o}we, O{\textquoteright}Brien,
  Laing, and Thompson]{MultidimOptics2018}
J.~Wang, S.~Paesani, Y.~Ding, R.~Santagati, P.~Skrzypczyk, A.~Salavrakos,
  J.~Tura, R.~Augusiak, L.~Man{\v c}inska, D.~Bacco, D.~Bonneau,
  J.~Silverstone, Q.~Gong, A.~Ac{\'\i}n, K.~Rottwitt, L.~Oxenl{\o}we,
  J.~O{\textquoteright}Brien, A.~Laing, and M.~Thompson.
\newblock Multidimensional quantum entanglement with large-scale integrated
  optics.
\newblock \emph{Science}, 360\penalty0 (6386):\penalty0 285--291,
  2018{\natexlab{b}}.
\newblock \doi{10.1126/science.aar7053}.

\end{thebibliography}

\appendix
\section{Marginal probabilities for partial photon outputs}\label{app:margprobs}

In this Appendix, we show how to compute the marginal probabilities necessary for \cref{thm:binnedBS}. This generalizes Lemma 1 of \cite{Clifford2017} by allowing for collision inputs. The proof technique, however, is slightly different from that of \cite{Clifford2017}, as we use the physical description of bosonic states via first quantization rather than combinatorial arguments. Therefore, we hope it might also make the original Lemma 1 of \cite{Clifford2017} more transparent for physicists. 

Throughout this section we denote by $\ket{\vec{T}}$ an $n$-boson, $m$-mode Dicke state. We can write it as
\begin{equation}\label{eq:1to2quantmap}
\ket{\vec{T}} \eqdef \binom{n}{t_1, t_2, \ldots t_m}^{\michal{1/2}} \proj{n}{m} \ket{\vec{r}},
\end{equation}
where $\proj{n}{m}$ is the projector onto the $n$-particle, $m$-mode symmetric subspace, and $\ket{\vec{r}} \eqdef \ket{r_1} \ldots \ket{r_n}$. The state $\ket{\vec{T}}$ is such that $t_i$ is the number of times state $\ket{i}$ appears in $\ket{\vec{r}}$. Conversely, $\vec{r}$ are the tuples as defined just above \cref{eq:perPDF}.

Let now $\ket{\vec{S}}$ and $\ket{\vec{T}}$ be the (not necessarily collision-free) input and output states to a linear-optical circuit $U$. In this formalism, we can write the corresponding transition probability as
\begin{equation}
P_U(\vec{S} \rightarrow \vec{T}) = \tr (\varphi_U \ketbra{\vec{S}}{\vec{S}} \varphi_U\dg \ketbra{\vec{T}}{\vec{T}}),
\end{equation}
where $\varphi_U$ is the unitary transformation induced by $U$ in the $n$-particle Hilbert space, i.e.\ $\varphi_U \eqdef U^{\otimes n}$. The above expression coincides with the standard definition of bosonic transition probabilities in terms of permanents of \cite{Aaronson2013a}.

Expanding the output state as in \cref{eq:1to2quantmap} we can write
\begin{equation}\label{eq:perPDF1st}
P_U(\vec{S} \rightarrow \vec{T}) = \binom{n}{t_1, t_2, \ldots t_m} \tr (\varphi_U \ketbra{\vec{S}}{\vec{S}} \varphi_U\dg \ketbra{\vec{r}}{\vec{r}}),
\end{equation}
where we have used the cyclic property of the trace and the fact that $\proj{n}{m}$ acts trivially on $\varphi_U \ketbra{\vec{S}}{\vec{S}} \varphi_U\dg$. \Cref{eq:perPDF1st} is similar to \cref{eq:perPDF}. It represents the fact that, due to the symmetries of the input state and the dynamics (or, equivalently, the symmetries of the permanent function), in order to sample from the output distribution over symmetric states (or multisets $\vec{z}$) we can equivalently sample from the distribution over tuples $\vec{r}$ given by
\begin{equation}\label{eq:perPDF1stb}
P_U(\vec{S} \rightarrow \vec{r}) = \tr (\varphi_U \ketbra{\vec{S}}{\vec{S}} \varphi_U\dg \ketbra{\vec{r}}{\vec{r}}),
\end{equation}
and reorder the outcome in nondecreasing manner.

\Cref{eq:perPDF1stb} gives us an expression for the full pmf over tuples $\vec{r}$. We now need to compute the marginal probabilities over leading subsequences of $\vec{r}$. To that end, let $\vec{r}_{l} = (r_1, r_2 \ldots r_l)$ denote the subsequence of the first $l$ elements of $\vec{r}$, for $l \leq n$, and let $\vec{r}_{l,C} = (r_l, r_{l+1} \ldots r_n)$ denote the remaining elements of the list over which we wish to marginalize. We can write the marginal pmf for that subsequence as
\begin{equation}\label{eq:perPDFpartial}
P_U(\vec{S} \rightarrow \vec{r}_l) = \tr \left[\varphi_U \ketbra{\vec{S}}{\vec{S}} \varphi_U\dg \left(\ketbra{\vec{r}_{l}}{\vec{r}_{l}} \otimes \sum_{\vec{r}_{l,C}} \ketbra{\vec{r}_{l,C}}{\vec{r}_{l,C}}\right)\right].
\end{equation}
The \michal{sum} inside the trace equals simply to the identity on the subspace of particles $l+1$ to $n$. We can thus write the above as 
\begin{align}
  P_U(\vec{S} \rightarrow \vec{r}_l) & = \tr \left[ \tr_{n-k} (\varphi_U \ketbra{\vec{S}}{\vec{S}} \varphi_U\dg) \ketbra{\vec{r}_{l}}{\vec{r}_{l}} \right] \notag\\
& =  \tr \left[ U^{\otimes l} (\tr_{n-l}  \ketbra{\vec{S}}{\vec{S}}) (U\dg)^{\otimes l} \ketbra{\vec{r}_{l}}{\vec{r}_{l}} \right], \label{eq:perPDFpartfinal}
\end{align}
where by $\tr_{n-l}$ we mean tracing over the last $n-l$ particles, and we have used the property that $\tr_{n-l} (U^{\otimes n} \rho (U\dg)^{\otimes n}) = U^{\otimes l} \tr_{n-l} (\rho) (U\dg)^{\otimes l}$, as discussed in \cite{Oszmaniec2018}.

The result we need already follows from \cref{eq:perPDFpartfinal}. To see this note that $(\tr_{n-l}  \ketbra{\vec{S}}{\vec{S}})$ is a  convex combination over all $l$-photon Dicke states compatible with choosing $l$ out of $n$ photons from input state $\ket{\vec{S}}$ i.e. states of the form
\begin{equation}\label{eq:comp}
\ket{\vec{S}-\vec{K}}=\ket{S_1-K_1,S_2-K_2,\ldots,S_m-K_m}\ ,
\end{equation}
where $\sum_{i=1}^m K_i = n-l$ and $K_i \leq S_i$. The probability that a particular state $\ket{\vec{S}-\vec{K}}$ appears in the decomposition 
of $(\tr_{n-l}  \ketbra{\vec{S}}{\vec{S}})$ is proportional to the multinomial coefficient $\binom{n-l}{K_1, K_2, \ldots K_m}$. Therefore, \cref{eq:perPDFpartfinal} satisfies properties (i) and (ii) laid out above \cref{eq:chain}. The fact that each probability in \cref{eq:perPDFpartfinal} can be computed efficiently for input state of types A and B follows from the discussion in the proof of \cref{thm:binnedBS}. 

\section{Upper bound on the cost of permanents involved in \cref{thm:binnedBS}}\label{app:CompCost}

The maximal runtime  associated to the computation of the individual permanent in the simulation algorithm given in the proof of \cref{thm:binnedBS} can be computed via the following expression.
 \begin{equation}\label{eq:marginals}
\tau_S \eqdef \max_{\michal{l}\in\lbrace 1,n\rbrace} \max_{S'\vdash S, |S'|=l}\max_{T,|T|=l} \tau_{S,T}\ ,
\end{equation}
where
\begin{equation}
\tau_{S,T}=\min\left(\prod_{i=1}^m(s_i+1),\prod_{j=1}^m(t_j+1) \right) \alpha_S \alpha_T\ 
\end{equation}
and the internal maximization is over all $l$-photon input configurations $\vec{S}'$ that are compatible with with the initial state $\vec{S}$ and over $l$-photon outputs $\vec{T}$. Let us note that $\tau_{S,T}$ is a non-decreasing function of individual components of both inputs $\vec{S}$ and $\vec{T}$. Therefore, for a specific number of \michal{particles} we have $\tau_S =\max_{T,|T|=n} \tau_{S,T}=\prod_{i=1}^m(s_i+1)\alpha_S n$. This concludes the proof of the upper bound from Eq.\eqref{eq:uppGLOB}.     
\end{document}